\documentclass[11pt]{article}

\usepackage{fullpage}
\usepackage{latexsym}
\usepackage{amssymb,amsfonts,amsmath,amsthm}
\usepackage{times}

\def\01{\{0,1\}}

\newcommand{\eps}{\varepsilon}

\newcommand{\DISJ}{\mbox{\rm DISJ}}

\newcommand{\Exp}{{\mathbb{E}}}
\newcommand{\ket}[1]{|#1\rangle}
\newcommand{\bra}[1]{\langle#1|}
\newcommand{\ketbra}[2]{|#1\rangle\langle#2|}
\newcommand{\norm}[1]{{\left\|{#1}\right\|}}
\newcommand{\normb}[1]{{\big\|{#1}\big\|}}
\newcommand{\inp}[2]{\langle{#1}|{#2}\rangle} 
\newcommand{\Tr}{\mbox{\rm Tr}}

\renewcommand{\Pr}{\mbox{\rm Pr}}
\newcommand{\R}{\mathbb{R}}
\newcommand{\C}{\mathbb{C}}
\newcommand{\M}{{\cal M}}
\newcommand{\half}{{\frac{1}{2}}}
\newcommand{\set}[1]{{\left\{#1\right\}}}
\newcommand{\ksubsets}{{[n] \choose k}}

\newcommand{\trnorm}[1]{\norm{#1}_{\rm tr}}
\newcommand{\trnormb}[1]{{\big\| #1 \big\|}_{\rm tr}}

\newtheorem{definition}{Definition}
\newtheorem{theorem}{Theorem}
\newtheorem{lemma}[theorem]{Lemma}

\newtheorem{fact}[theorem]{Fact}

\newcommand{\mnote}[1]{}

\renewcommand{\qed}{\hfill{\rule{2mm}{2mm}}}
\renewenvironment{proof}[1][]{\begin{trivlist}
\item[\hspace{\labelsep}{\bf\noindent Proof#1:\/}] }{\qed\end{trivlist}}

\begin{document}

\title{A Hypercontractive Inequality for Matrix-Valued Functions\\ with Applications to Quantum Computing and LDCs}
\author{
Avraham Ben-Aroya\thanks{School of Computer Science, Tel-Aviv University, Tel-Aviv 69978, Israel.
   Supported by the Adams Fellowship Program of the Israel Academy of Sciences and Humanities, by the Israel Science Foundation,
   and by the European Commission under the
Integrated Project QAP funded by the IST directorate as Contract Number 015848.
}
\and
Oded Regev\thanks{School of Computer Science, Tel-Aviv University, Tel-Aviv 69978, Israel. Supported
   by the Binational Science Foundation, by the Israel Science Foundation, and
   by the European Commission under the Integrated Project QAP funded by the IST directorate as Contract Number 015848.
}
\and
Ronald de Wolf\thanks{Centrum voor Wiskunde en Informatica (CWI), Amsterdam, The Netherlands. Supported by a Veni grant from the Netherlands Organization for Scientific Research
(NWO) and also partially supported by the European Commission under the Integrated Project QAP funded by the IST directorate as Contract Number 015848.}}
\date{}
\maketitle
\thispagestyle{empty}

\begin{abstract}
The Bonami-Beckner hypercontractive inequality is a powerful tool in Fourier analysis
of real-valued functions on the Boolean cube.
In this paper we present a version of this inequality for \emph{matrix-valued} functions on the Boolean cube.
Its proof is based on a powerful inequality by Ball, Carlen, and Lieb.
We also present a number of applications.
First, we analyze maps that encode $n$ classical
bits into $m$ qubits, in such a way that each set of $k$ bits can be recovered with some probability by
an appropriate measurement on the quantum encoding;
we show that if $m<0.7 n$, then the success probability is exponentially small in~$k$.
This result may be viewed as a direct product version of Nayak's quantum random access code bound.
It in turn implies strong direct product theorems for the one-way quantum communication complexity
of Disjointness and other problems.
Second, we prove that error-correcting codes that are locally decodable with 2 queries require
length exponential in the length of the encoded string.
This gives what is arguably the first ``non-quantum'' proof of a result originally derived by Kerenidis and de Wolf using quantum information theory,
and answers a question by Trevisan.
\end{abstract}

\newpage
\setcounter{page}{1}

\section{Introduction}

\subsection{A hypercontractive inequality for matrix-valued functions}

Fourier analysis of real-valued functions on the Boolean cube has been
widely used in the theory of computing. Applications include analyzing the influence of variables
on Boolean functions~\cite{kkl:influence}, 
probabilistically-checkable proofs and associated hardness of approximation~\cite{hastad:optimalj},
analysis of threshold phenomena~\cite{kalai&safra:threshold},
noise stability~\cite{moo:noisestability,odonnell:thesis}, 
voting schemes~\cite{odonnell:survey},
learning under the uniform distribution
\cite{lmn:learnability,mansour:dnfuniform,jackson:dnf,mos:learningjuntas},
communication complexity~\cite{raz:ccfourier,klauck:qcclower,gkkrw:1way}, etc.

One of the main technical tools in this area is a hypercontractive inequality that is sometimes
called the \emph{Bonami-Beckner inequality}~\cite{Bon70,Bec75}, though its history would also justify other names
(see Lecture~16 of~\cite{odonnell:notes} for some background and history).
For a fixed $\rho\in[0,1]$, consider the linear operator $T_\rho$ on the space of all
functions $f:\01^n\rightarrow\R$ defined by
$$
(T_\rho(f))(x)= \Exp_y [f(y)],
$$
where the expectation is taken over $y$ obtained from $x$ by negating each bit independently with probability $(1-\rho)/2$.
In other words, the value of $T_\rho(f)$ at a point $x$ is obtained by averaging
the values of $f$ over a certain neighborhood of $x$.
One important property of $T_\rho$ for $\rho<1$ is that it has a ``smoothing" effect:
any ``high peaks" present in $f$ are smoothed out in $T_\rho (f)$.
The hypercontractive inequality formalizes this intuition.
To state it precisely, define the $p$-norm of a function $f$
by $\norm{f}_p=(\frac{1}{2^n}\sum_x |f(x)|^p)^{1/p}$. It is not difficult to prove that the norm
is nondecreasing with $p$. Also, the higher $p$ is, the more sensitive the norm becomes to peaks in the
function~$f$. The hypercontractive inequality says that for certain $q > p$, the $q$-norm of $T_\rho (f)$ is
upper bounded by the $p$-norm of $f$. This exactly captures the intuition that $T_\rho( f)$ is a smoothed
version of $f$: even though we are considering a higher norm, the norm does not increase.
More precisely, the hypercontractive inequality says that as long as $1\leq p\leq q$ and
$\rho\leq\sqrt{(p-1)/(q-1)}$, we have
\begin{equation}\label{eq:hypercontractivegeneral}
\norm{T_\rho (f)}_q\leq\norm{f}_p.
\end{equation}

The most interesting case for us is when $q=2$, since in this case
one can view the inequality as a statement about the Fourier coefficients of $f$,
as we describe next. Let us first recall some basic definitions from Fourier analysis.
For every $S\subseteq [n]$ (which by some abuse of notation we will also view as an $n$-bit string)
and $x \in \{0,1\}^n$, define $\chi_S(x)=(-1)^{x\cdot S}$ to be the parity of the bits of $x$ indexed by $S$.
The \emph{Fourier transform} of a function $f:\01^n\to\R$ is the function $\widehat{f}:\01^n\to\R$ defined by
$$
\widehat{f}(S)=\frac{1}{2^n}\sum_{x\in\01^n}f(x)\chi_S(x).
$$
The values $\widehat{f}(S)$ are called the \emph{Fourier coefficients} of $f$.
The coefficient $\widehat{f}(S)$ may be viewed as measuring the correlation between $f$
and the parity function $\chi_S$. Since the functions $\chi_S$ form an orthonormal basis
of the space of all functions from $\01^n$ to $\R$, we can express $f$ in terms
of its Fourier coefficients as
\begin{align}\label{eq:invfourier}
f=\sum_{S\subseteq [n]}\widehat{f}(S)\chi_S.
\end{align}
Using the same reasoning we obtain Parseval's identity,
$$
\norm{f}_2=\left( \sum_{S\subseteq [n]}\widehat{f}(S)^2 \right)^{1/2}.
$$
The operator $T_\rho$ has a particularly
elegant description in terms of the Fourier coefficients. Namely, it simply
multiplies each Fourier coefficient $\widehat{f}(S)$ by a factor of $\rho^{|S|}$:
$$
T_\rho(f)=\sum_{S \subseteq [n]}\rho^{|S|} \widehat{f}(S)\chi_S.
$$
The higher $|S|$ is, the stronger the Fourier coefficient $\widehat{f}(S)$ is ``attenuated'' by $T_\rho$.
Using Parseval's identity, we can now write the hypercontractive inequality~\eqref{eq:hypercontractivegeneral}
for the case $q=2$ as follows. For every $p \in [1,2]$,
\begin{equation}\label{eqbonamibeckner}
\Bigg( \sum_{S\subseteq[n]} (p-1)^{|S|}\widehat{f}(S)^2 \Bigg) ^ {1/2} \leq
  \Bigg(\frac{1}{2^n}\sum_{x\in\01^n} |f(x)|^p\Bigg)^{1/p}.
\end{equation}
This gives an upper bound on a weighted sum of the squared Fourier coefficients of $f$,
where each coefficient is attenuated by a factor $(p-1)^{|S|}$.
We are interested in generalizing this hypercontractive inequality to \emph{matrix-valued} functions.
Let $\M$ be the space of $d\times d$ complex matrices and suppose we have a function $f:\01^n\to\M$.
For example, a natural scenario where this arises is in quantum information theory,
if we assign to every $x\in\01^n$ some $m$-qubit \emph{density matrix} $f(x)$ (so $d=2^m$).
We define the Fourier transform $\widehat{f}$ of a matrix-valued function $f$ exactly as before:
$$
\widehat{f}(S)=\frac{1}{2^n}\sum_{x\in\01^n}f(x)\chi_S(x).
$$
The Fourier coefficients $\widehat{f}(S)$ are now also $d\times d$ matrices.
An equivalent definition is by applying the standard Fourier transform
to each $i,j$-entry separately: $\widehat{f}(S)_{ij}=\widehat{f(\cdot)_{ij}}(S)$.
This extension of the Fourier transform to matrix-valued functions is quite
natural, and has also been used in, e.g.,~\cite{nayak:quantumwalk,fehr&schaffner:extraction}.

Our main tool, which we prove in Section~\ref{secproofmatrixbeckner}, is an extension
of the hypercontractive inequality to matrix-valued functions.
For $M\in\M$ with singular values $\sigma_1,\ldots,\sigma_d$, we define its (normalized Schatten) $p$-norm
as $\norm{M}_p=(\frac{1}{d}\sum_{i=1}^d\sigma_i^p)^{1/p}$.

\begin{theorem}\label{thm:main-thm}
For every $f:\01^n \to \M$ and $1 \le p \le 2$,
$$
\Bigg( \sum_{S \subseteq [n]} (p-1)^{|S|} \normb{\widehat{f}(S)}_p^2 \Bigg)^{1/2} \le
   \Bigg( \frac{1}{2^{n}} \sum_{x \in \01^n} \norm{f(x)}_p^p \Bigg)^{1/p}.
$$
\end{theorem}

This is the analogue of Eq.~\eqref{eqbonamibeckner} for matrix-valued functions, with $p$-norms replacing absolute values.
The case $n=1$ can be seen as a geometrical statement that extends
the familiar parallelogram law in Euclidean geometry and is closely related
to the notion of uniform convexity. This case
was first proven for certain values of $p$ by Tomczak-Jaegermann \cite{Tomczak74}
and then in full generality by Ball, Carlen, and Lieb~\cite{bcl:tracenorms}.
Among its applications are the work of Carlen and Lieb on fermion fields \cite{CarlenLieb},
and the more recent work of Lee and Naor on metric embeddings \cite{LeeNaor04}.

To the best of our knowledge, the general case $n \ge 1$ has not
appeared before.\footnote{A different generalization
of the Bonami-Beckner inequality was given by Borell~\cite{Borell79}. His
generalization, however, is an easy corollary of the Bonami-Beckner inequality
and is therefore relatively weak (although it does apply to any Banach space, and
not just to the space of matrices with the Schatten $p$-norm).}
Its proof is not difficult, and follows by induction on $n$, similar to the
proof of the usual hypercontractive inequality.\footnote{We remark
that Carlen and Lieb's proof in~\cite{CarlenLieb} also uses induction
and has some superficial resemblance to the proof given here. Their induction,
however, is on the \emph{dimension} of the matrices (or more precisely,
the number of fermions), and moreover leads to an entirely different inequality.}
Although one might justly regard Theorem~\ref{thm:main-thm} as a ``standard'' corollary of
the result by Ball, Carlen, and Lieb, such ``tensorized inequalities''
tend to be extremely useful (see, e.g., \cite{Bobkov97,Gross75})
and we believe that the matrix-valued hypercontractive inequality
will have more applications in the future.

\subsection{Application: $k$-out-of-$n$ random access codes}\label{ssecintrokrac}

Our main application of Theorem~\ref{thm:main-thm} is for the following information-theoretic problem.
Suppose we want to encode an $n$-bit string $x$ into $m$ bits or qubits, in such a way that
for any set $S\subseteq[n]$ of $k$ indices, the $k$-bit substring $x_S$ can
be recovered with probability at least $p$ by making an appropriate measurement on the encoding.
We are allowed to use probabilistic encodings here, so the encoding need not be a function mapping
$x$ to a fixed classical string or a fixed quantum pure state.
We will call such encodings \emph{$k$-out-of-$n$ random access codes}, since they allow us to access any
set of $k$ out of $n$ bits. As far as we know, for $k>1$ neither the classical nor the quantum case
has been studied before. Here we focus on the quantum case, because our lower bounds for quantum encodings
of course also apply to classical encodings.

We are interested in the tradeoff between the length $m$ of the quantum random access code,
and the success probability $p$.
Clearly, if $m\geq n$ then we can just use the identity encoding to obtain $p=1$.
If $m<n$ then by Holevo's theorem~\cite{holevo} our encoding will be ``lossy'', and $p$ will be less than~1.
The case $k=1$ was first studied by Ambainis et al.~\cite{antv:dense},
who showed that if $p$ is bounded away from 1/2, then $m=\Omega(n/\log n)$.
Nayak~\cite{nayak:qfa} subsequently strengthened this bound to $m\geq (1-H(p))n$,
where $H(\cdot)$ is the binary entropy function.
This bound is optimal up to an additive $\log n$ term both for classical and quantum encodings.
The intuition of Nayak's proof is that, for average $i$, the encoding only contains $m/n<1$ bits
of information about the bit $x_i$, which limits our ability to predict $x_i$ given the encoding.

Now suppose that $k>1$, and $m$ is much smaller than $n$. Clearly, for predicting one specific
bit $x_i$, with $i$ uniformly chosen, Nayak's result applies, and we will have a success probability
that is bounded away from~1. But intuitively this should apply to each of the $k$ bits that we need
to predict.  Moreover, these $k$ success probabilities should not be very correlated,
so we expect an overall success probability that is exponentially small in $k$.
Nayak's proof does not generalize to the case $k\gg 1$ (or at least, we do not know how to do it).
The reason it fails is the following.
Suppose we probabilistically encode $x\in\01^n$ as follows: with probability 1/4 our encoding is $x$ itself,
and with probability 3/4 our encoding is the empty string.
Then the average length of the output (and hence the entropy or amount of information in the encoding)
is only $n/4$ bits, or 1/4 bit for an average $x_i$. Yet from this encoding one can predict
\emph{all of} $x$ with success probability $1/4$!
Hence, if we want to prove our intuition, we should make use of the fact that the encoding
is always confined to a $2^m$-dimensional space (a property which the above example lacks).
Arguments based on von Neumann entropy, such as the one of~\cite{nayak:qfa}, do not seem
capable of capturing this condition (however, a \emph{min-entropy} argument recently
enabled K\"{o}nig and Renner to prove a closely related but incomparable result, see below).
The new hypercontractive inequality offers an alternative approach---in fact the only
alternative approach to entropy-based methods that we are aware of in quantum information.
Applying the inequality to the matrix-valued function that gives the encoding implies
$p\leq 2^{-\Omega(k)}$ if $m\ll n$.  More precisely:

\begin{theorem}\label{thm:strongqrac}
For any $\eta> 2 \ln 2$ there exists a constant $C_\eta$ such
that if $n/k$ is large enough then for any $k$-out-of-$n$ quantum random access code on $m$ qubits,
the success probability satisfies
$$
p\leq C_\eta \left(\frac{1}{2}+\frac{1}{2}\sqrt{\frac{\eta m}{n}}\right)^k.
$$
\end{theorem}

In particular, the success probability is exponentially small in $k$ if $m/n<1/(2\ln 2)\approx 0.721$.
Notice that for very small $m/n$ the bound on $p$ gets close to $2^{-k}$, which is what one gets
by guessing the $k$-bit answer randomly.
We also obtain bounds if $k$ is close to $n$, but these are a bit harder to state.
We believe that the theorem can be extended to the case that $m/n> 1/(2\ln 2)$, although
proving this would probably require a strengthening of the inequality by Ball, Carlen,
and Lieb. Luckily, in all our applications we are free to choose a small enough $m$.
Finally, we note that in contrast to Nayak's approach, our proof
does not use the strong subadditivity of von Neumann entropy.

\medskip

\paragraph{The classical case.}
We now give a few comments regarding the special case of
classical (probabilistic) $m$-bit encodings.
First, in this case the encodings are represented by diagonal matrices. For such matrices, the base
case $n=1$ of Theorem \ref{thm:main-thm} can be derived directly from the Bonami-Beckner
inequality, without requiring the full strength of the Ball-Carlen-Lieb inequality
(see \cite{bcl:tracenorms} for details).
Alternatively, one can derive Theorem~\ref{thm:strongqrac} in the classical case
directly from the Bonami-Beckner inequality by conditioning on a fixed $m$-bit string of the
encoding (this step is already impossible in the quantum case) and then analyzing
the resulting distribution on $\01^n$. This proof is very similar to
the one we give in Section~\ref{seckrac} (and in fact slightly less elegant
due to the conditioning step) and we therefore omit the details.

Interestingly, in the classical case there is a simpler argument that avoids Bonami-Beckner altogether.
This argument was used in~\cite{viola&wigderson:multipartypj} and was communicated
to us by the authors of that paper. We briefly sketch it here.
Suppose we have a classical (possibly randomized) $m$-bit encoding that allows
to recover any $k$-bit set with probability at least $p$ using a (possibly randomized)
decoder. By Yao's minimax principle, there is a way to fix the randomness in both the
encoding and decoding procedures, such that the probability of succeeding in recovering
all $k$ bits of a randomly chosen $k$-set from an encoding of a uniformly random $x\in\01^n$
is at least $p$. So now we have deterministic encoding and decoding,
but there is still randomness in the input $x$.
Call an $x$ ``good'' if the probability of the decoding procedure being successful on a random $k$-tuple is at least $p/2$
(given the $m$-bit encoding of that $x$).
By Markov's inequality, at least a $p/2$-fraction of the inputs $x$ are good. 
Now consider the following experiment. Given the encoding of a uniform $x$,
we take $\ell=100n/k$ uniformly and independently chosen $k$-sets
and apply the decoding procedure to all of them. We then output an $n$-bit string
with the ``union'' of all the answers we received (if we received multiple contradictory answers for the same bit, 
we can put either answer there), and random bits for the positions that are not in the union.
With probability $p/2$, $x$ is good. Conditioned on this, with
probability at least $(p/2)^\ell$ all our decodings are correct. Moreover,
except with probability $2^{-\Omega(n)}$, the union of our $\ell$ $k$-sets
is of size at least $0.9n$. The probability of guessing the remaining $n/10$ bits right is $2^{-n/10}$.
Therefore the probability of successfully recovering all of $x$ is at least $(p/2)\cdot((p/2)^\ell - 2^{-\Omega(n)})\cdot 2^{-n/10}$.
A simple counting argument shows that this is impossible unless $p\leq 2^{-\Omega(k)}$ or $m$ is close to $n$.
This argument does not work for quantum encodings, of course, because these cannot just be reused
(a quantum measurement changes the state).

\medskip

\paragraph{The K\"{o}nig-Renner result.}
Independently but subsequent to our work (which first appeared on the arxiv preprint server in May~2007),
K\"{o}nig and Renner~\cite{koenig&renner:samplingminentropy} recently used sophisticated quantum
information theoretic arguments to show a result with a similar flavor to ours.
Each of the results is tuned for different scenarios. In particular, the results are incomparable,
and our applications to direct product theorems do not follow from their result,
nor do their applications follow from our result.
We briefly describe their result and explain the distinction between the two.

Let $X=X_1,\ldots,X_n$ be classical random variables, not necessarily uniformly distributed or even independent.
Suppose that each $X_i \in\01^b$.
Suppose further that the ``smooth min-entropy of $X$ relative to a quantum state $\rho$'' is at least
some number $h$ (see~\cite{koenig&renner:samplingminentropy} for the precise definitions, which are quite technical).
If we randomly pick $r$ distinct indices $i_1,\ldots,i_r$,
then intuitively the smooth min-entropy of $X'=X_{i_1},\ldots,X_{i_r}$ relative to $\rho$ should not be much smaller than $hr/n$.
K\"{o}nig and Renner show that if $b$ is larger than $n/r$ then this is indeed the case, except with probability exponentially small in~$r$.
Note that they are picking $b$-bit blocks $X_{i_1},\ldots,X_{i_r}$ instead of individual bits, but this
can also be viewed as picking (not quite uniformly) $k=rb$ bits from a string of $nb$ bits.

On the one hand, the constants in their bounds are essentially optimal,
while ours are a factor $2\ln 2$ off from what we expect they should be.
Also, while they need very few assumptions on the random variables $X_1,\ldots,X_n$ and on
the quantum encoding, we assume the random variables are uniformly distributed bits,
and our quantum encoding is confined to a $2^m$-dimensional space.
We can in fact slightly relax both the assumption on the input and the encoding, but do not discuss
these relaxations since they are of less interest to us.
Finally, their result still works if the indices $i_1,\ldots,i_r$ are not sampled uniformly,
but are sampled in some randomness-efficient way.
This allows them to obtain efficient key-agreement schemes in a cryptographic model
where the adversary can only store a bounded number of quantum bits.

On the other hand, our result works even if only a small number of bits is sampled,
while theirs only kicks in when the number of bits being sampled
($k=rb$) is at least the square-root of the total number of bits $nb$.
This is not very explicit in their paper, but can be seen by observing that the parameter $\kappa=n/(rb)$
on page~8 and in Corollary~6.19 needs to be at most a constant (whence the assumption that $b$ is larger than $n/r$).
So the total number of bits is $nb=O(r b^2)=O(r^2 b^2)=O(k^2)$.
Since we are interested in small as well as large $k$, this limitation of their approach is significant.
A final distinction between the results is in the length of the proof.
While the information-theoretic intuition in their paper is clear and well-explained,
the details get to be quite technical, resulting in a proof which is significantly longer than ours.

\subsection{Application: Direct product theorem for one-way quantum communication complexity}
Our result for $k$-out-of-$n$ random access codes has the flavor of a direct product theorem:
the success probability of performing a certain task on $k$ instances (i.e.,~$k$ distinct indices)
goes down exponentially with~$k$. In Section~\ref{secdirectproddisj}, we
use this to prove a new strong direct product theorem for one-way communication complexity.

Consider the 2-party Disjointness function:
Alice receives input $x\in\01^n$, Bob receives input $y\in\01^n$, and they want to determine
whether the sets represented by their inputs are disjoint, i.e.~whether $x_i y_i=0$ for all $i\in[n]$.
They want to do this while communicating as few qubits as possible (allowing some small error probability, say 1/3).
We can either consider one-way protocols, where Alice sends one message to Bob who then computes the output;
or two-way protocols, which are interactive.
The quantum communication complexity of Disjointness is fairly well understood:
it is $\Theta(n)$ qubits for one-way protocols~\cite{buhrman&wolf:qcclower},
and $\Theta(\sqrt{n})$ qubits for two-way protocols
\cite{BuhrmanCleveWigderson98,hoyer&wolf:disjeq,aaronson&ambainis:search,razborov:qdisj}.

Now consider the case of $k$ independent instances: Alice receives inputs $x_1,\ldots,x_k$ (each of $n$ bits),
Bob receives $y_1,\ldots,y_k$, and their goal is to compute all $k$ bits $\DISJ_n(x_1,y_1),\ldots,\DISJ_n(x_k,y_k)$.
Klauck et al.~\cite{ksw:dpt} proved an optimal direct product theorem for \emph{two-way}
quantum communication: every protocol that communicates fewer than $\alpha k\sqrt{n}$ qubits
(for some small constant $\alpha>0$) will have a success probability that is exponentially small in $k$.
Surprisingly, prior to our work no strong direct product theorem was known for the usually simpler case
of \emph{one-way} communication---not even for \emph{classical} one-way communication.%
\footnote{Recently and independently of our work, Jain et al.~\cite{jkn:subdistribution} did manage to prove such a
direct product theorem for classical one-way communication, based on information-theoretic techniques.}
In Section~\ref{secdirectproddisj} we derive such a theorem from our $k$-out-of-$n$ random access code lower bound:
if $\eta> 2\ln 2$, then every one-way quantum protocol that sends fewer than $kn/\eta$ qubits
will have success probability at most $2^{-\Omega(k)}$.

These results can straightforwardly be generalized to get a bound for all functions in terms
of their \emph{VC-dimension}. If $f$ has VC-dimension $d$, then any one-way quantum protocol for
computing $k$ independent copies of $f$ that sends $kd/\eta$ qubits, has success probability $2^{-\Omega(k)}$.
For simplicity, Section~\ref{secdirectproddisj} only presents the case of Disjointness.
Finally, by the work of Beame et al.~\cite{bpsw:corruptionj}, such direct product theorems imply lower bounds on \emph{3-party} protocols
where the first party sends only one message. We elaborate on this in Appendix~\ref{app3partynof}.

\subsection{Application: Locally decodable codes}

A locally decodable error-correcting code (LDC) $C:\01^n\rightarrow\01^N$ encodes $n$ bits into $N$ bits,
in such a way that each encoded bit can be recovered from a noisy codeword by a randomized
decoder that queries only a small number $q$ of bit-positions in that codeword.
Such codes have applications in a variety of different complexity-theoretic and cryptographic settings;
see for instance Trevisan's survey and the references therein~\cite{trevisan:eccsurvey}.
The main theoretical issue in LDCs is the tradeoff between $q$ and $N$.
The best known constructions of LDCs with constant $q$ have a length $N$ that is sub-exponential in $n$
but still superpolynomial~\cite{cgks:pir,bikr:improvedpir,yekhanin:3ldc}.
On the other hand, the only superpolynomial \emph{lower} bound known for general LDCs
is the tight bound $N=2^{\Omega(n)}$ for $q=2$ due to Kerenidis and de Wolf~\cite{kerenidis&wolf:qldcj}
(generalizing an earlier exponential lower bound for \emph{linear} codes by~\cite{gkst:lowerpir}).
Rather surprisingly, the proof of~\cite{kerenidis&wolf:qldcj} relied heavily on techniques from
quantum information theory: despite being a result purely about classical codes and classical decoders,
the quantum perspective was crucial for their proof. In particular, they show that the two queries
of a classical decoder can be replaced by one quantum query, then they turn this quantum
query into a random access code for the encoded string $x$, and finally invoke Nayak's
lower bound for quantum random access codes.

In Section~\ref{secldc} we reprove an exponential lower bound on $N$ for the case $q=2$
without invoking any quantum information theory:
we just use classical reductions, matrix analysis, and the hypercontractive inequality for matrix-valued functions.
Hence it is a classical (non-quantum) proof as asked for by Trevisan~\cite[Open question~3 in Section~3.6]{trevisan:eccsurvey}.%
\footnote{Alex Samorodnitsky has been developing a classical proof along similar lines in the past two years.
However, as he told us at the time of writing~\cite{samorodnitsky:ldccomm}, his proof is still incomplete.}
It should be noted that this new proof is still quite close in spirit (though not terminology) to the quantum proof
of~\cite{kerenidis&wolf:qldcj}. This is not too surprising given the fact that the proof of~\cite{kerenidis&wolf:qldcj}
uses Nayak's lower bound on random access codes, generalizations of which follow from the hypercontractive inequality.
We discuss the similarities and differences between the two proofs in Section~\ref{secldc}.

We feel the merit of this new approach is not so much in giving a partly new proof of the known lower
bound on 2-query LDCs, but in its potential application to codes with more than~2 queries.
Recently Yekhanin~\cite{yekhanin:3ldc} constructed 3-query LDCs with $N=2^{O(n^{1/32582657})}$
(and $N=2^{n^{O(1/\log\log n)}}$ for infinitely many $n$ if there exist infinitely many Mersenne primes).
For $q=3$, the best known lower bounds on $N$ are slightly less than $n^2$~\cite{katz&trevisan:ldc,kerenidis&wolf:qldcj,woodruff:ldclower}.
Despite considerable effort, this gap still looms large. Our hope is that our approach can be generalized to
3 or more queries. Specifically, what we would need is a generalization of tensors of rank~2 (i.e.,~matrices)
to tensors of rank $q$; an appropriate tensor norm;
and a generalization of the hypercontractive inequality from matrix-valued to tensor-valued functions.
Some preliminary progress towards this goal was obtained in~\cite{HavivR08}.

\section{Preliminaries}

\paragraph{Norms:}
Recall that we define the $p$-norm of a $d$-dimensional vector $v$ by
$$
\norm{v}_p = \left( \frac{1}{d} \sum_{i = 1}^d |v_i|^p \right)^{1/p}.
$$
We extend this to matrices by
defining the (normalized Schatten) $p$-norm of a matrix $A \in \C^{d \times d}$ as
$$
\norm{A}_p = \left(  \frac{1}{d} \Tr |A|^p \right)^{1/p}.
$$
This is equivalent to the $p$-norm of the vector of singular values of $A$.
For diagonal matrices this definition coincides with the one
for vectors. For convenience we defined all norms to be under the normalized
counting measure, even though for matrices this is nonstandard. The advantage of the
normalized norm is that it is nondecreasing with $p$.
We also define the \emph{trace norm} $\trnorm{A}$ of a matrix $A$ as the sum of its singular values,
hence we have $\trnorm{A} = d \norm{A}_1$ for any $d \times d$ matrix $A$.

\paragraph{Quantum states:}
An $m$-qubit \emph{pure state}
is a superposition $\ket{\phi}=\sum_{z\in\01^m}\alpha_z\ket{z}$ over all classical $m$-bit states. The $\alpha_z$'s are
complex numbers called \emph{amplitudes}, and $\sum_z|\alpha_z|^2=1$. Hence a pure state $\ket{\phi}$ is a unit vector
in $\mathbb{C}^{2^m}$.
Its complex conjugate (a row vector with entries conjugated) is denoted $\bra{\phi}$. The inner
product between $\ket{\phi}=\sum_z\alpha_z\ket{z}$ and $\ket{\psi}=\sum_z\beta_z\ket{z}$ is the dot product
$\bra{\phi}\cdot\ket{\psi}=\inp{\phi}{\psi}=\sum_z\alpha_z^*\beta_z$.
An $m$-qubit \emph{mixed state} (or \emph{density matrix}) $\rho=\sum_i p_i\ketbra{\phi_i}{\phi_i}$ corresponds to a
probability distribution over $m$-qubit pure states, where $\ket{\phi_i}$ is given with probability $p_i$.
The eigenvalues $\lambda_1,\ldots,\lambda_d$ of $\rho$ are non-negative reals that sum to~1,
so they form a probability distribution.
If $\rho$ is pure then one eigenvalue is 1 while all others are 0.
Hence for any $p\geq 1$, the maximal $p$-norm is achieved by pure states:
\begin{equation}\label{clm:dm-max-norm}
\norm{\rho}_p^p = \frac{1}{d} \sum_{i=1}^{d} \lambda_i^p \le \frac{1}{d}  \sum_{i=1}^{d} \lambda_i = \frac{1}{d}.
\end{equation}

A $k$-outcome \emph{positive operator-valued measurement} (POVM) is given by $k$ positive semidefinite operators $E_1,\ldots,E_k$
with the property that $\sum_{i=1}^k E_i=I$. When this POVM is applied to a mixed state $\rho$, the probability of the
$i$th outcome is given by the trace $\Tr(E_i\rho)$. The following well known fact gives the close relationship
between trace distance and distinguishability of density matrices:
\begin{fact}\label{trnorm-bias-fact}
The best possible measurement to distinguish two density matrices $\rho_0$ and $\rho_1$
has bias $\frac{1}{2}\trnorm{\rho_0 - \rho_1}$.
\end{fact}
Here ``bias'' is defined as twice the success probability, minus 1.
We refer to Nielsen and Chuang~\cite{nielsen&chuang:qc} for more details.

\section{The hypercontractive inequality for matrix-valued functions}\label{secproofmatrixbeckner}

Here we prove Theorem~\ref{thm:main-thm}. The proof relies on the following
powerful inequality by Ball et al.~\cite{bcl:tracenorms} (they
state this inequality for the usual unnormalized Schatten $p$-norm,
but both statements are clearly equivalent).

\begin{lemma}\label{lem:BCL-ineq} (\cite[Theorem 1]{bcl:tracenorms})
For any matrices $A,B$ and any $1 \le p \le 2$, it
holds that
$$
\left(\norm{\frac{A+B}{2}}^2_p + (p-1)\, \norm{\frac{A-B}{2}}^2_p \right)^{1/2} \le
   \left( \frac{\norm{A}^p_p + \norm{B}^p_p}{2} \right)^{1/p}.
$$
\end{lemma}

\newtheorem*{thma}{Theorem~\ref{thm:main-thm}}
\begin{thma}
For any $f:\01^n \to \M$ and for any $1 \le p \le 2$,
$$
\Bigg( \sum_{S \subseteq [n]} (p-1)^{|S|} \normb{\widehat{f}(S)}_p^2 \Bigg)^{1/2} \le
   \Bigg( \frac{1}{2^{n}} \sum_{x \in \01^n} \norm{f(x)}_p^p \Bigg)^{1/p}.
$$
\end{thma}

\begin{proof}
By induction.
The case $n=1$ follows from Lemma~\ref{lem:BCL-ineq} by setting $A = f(0)$ and $B = f(1)$,
and noting that $(A+B)/2$ and $(A-B)/2$ are exactly the Fourier coefficients $\widehat{f}(0)$ and $\widehat{f}(1)$.

We now assume the lemma holds for $n$ and prove it for $n+1$.
Let $f:\01^{n+1} \to \M$ be some matrix-valued function.
For $i \in \01$, let $g_i = f|_{x_{n+1} = i}$
be the function obtained by fixing the last input bit of $f$ to $i$.
We apply the induction hypothesis on $g_0$ and $g_1$ to obtain
\begin{align*}
\left( \sum_{S \subseteq [n]} (p-1)^{|S|} \norm{\widehat{g_0}(S)}_p^2 \right)^{1/2} &\le
  \left( \frac{1}{2^{n}} \sum_{x \in \01^n} \norm{g_0(x)}_p^p \right)^{1/p} \\
\left( \sum_{S \subseteq [n]} (p-1)^{|S|} \norm{\widehat{g_1}(S)}_p^2 \right)^{1/2} &\le
  \left( \frac{1}{2^{n}} \sum_{x \in \01^n} \norm{g_1(x)}_p^p \right)^{1/p}.
\end{align*}
Take the $L_p$ average of these two inequalities:
raise each to the $p$th power, average them and take the $p$th root. We get
\begin{align}
  \label{induction-ineq}
 \left( \half \sum_{i \in \01} \left( \sum_{S \subseteq [n]} (p-1)^{|S|} \norm{\widehat{g_i}(S)}_p^2 \right)^{p/2} \right)^{1/p} &\le
 \left( \frac{1}{2^{n+1}} \sum_{x \in \01^n} \left(\norm{g_0(x)}_p^p + \norm{g_1(x)}_p^p\right) \right)^{1/p} \\
 &= \left( \frac{1}{2^{n+1}} \sum_{x \in \01^{n+1}} \norm{f(x)}_p^p \right)^{1/p}. \nonumber
\end{align}
The right-hand side is the expression we wish to lower bound.
To bound the left-hand side, we need the following inequality
(to get a sense of why this holds, consider the case where $q_1=1$ and $q_2=\infty$).

\begin{lemma}[Minkowski's inequality, {\cite[Theorem 26]{HardyLP52}}]\label{lem:minkowski}
For any $r_1 \times r_2$ matrix whose rows are given by $u_1,\ldots,u_{r_1}$ and whose columns
are given by $v_1,\ldots,v_{r_2}$, and any $1 \le q_1 < q_2 \le \infty$,
$$
\norm{\left(\norm{v_1}_{q_2},\ldots,\norm{v_{r_2}}_{q_2}\right)}_{q_1} \ge
\norm{\left(\norm{u_1}_{q_1},\ldots,\norm{u_{r_1}}_{q_1}\right)}_{q_2},
$$
i.e., the value obtained by taking the $q_2$-norm of each column and then
taking the $q_1$-norm of the results, is at least that
obtained by first taking the $q_1$-norm of each row and then taking
the $q_2$-norm of the results.
\end{lemma}

Consider now the $2^n \times 2$ matrix whose entries are given by
$$
c_{S,i} = 2^{n/2} \norm{(p-1)^{|S|/2} \widehat{g_i}(S)}_p
$$
where $i \in \01$ and $S \subseteq [n]$.
The left-hand side of~\eqref{induction-ineq} is then
\begin{align*}
 \left( \half \sum_{i\in\{0,1\}} \left( \frac{1}{2^{n}} \sum_{S \subseteq [n]} c_{S,i}^{2} \right)^{p/2} \right)^{1/p} &\ge
 \left( \frac{1}{2^{n}} \sum_{S \subseteq [n]} \left( \half \sum_{i \in \{0,1\}} c_{S,i}^{p} \right)^{2/p} \right)^{1/2} \\
 &=
 \left( \sum_{S \subseteq [n]} (p-1)^{|S|}  \left(\frac{\norm{\widehat{g_0}(S)}_p^p + \norm{\widehat{g_1}(S)}_p^p }{2} \right)^{2/p}
 \right)^{1/2},
\end{align*}
where the inequality follows from Lemma~\ref{lem:minkowski} with $q_1=p$, $q_2=2$.
We now apply Lemma~\ref{lem:BCL-ineq} to deduce that the above is lower bounded by
$$
\left( \sum_{S \subseteq [n]} (p-1)^{|S|}
    \left(  \norm{\frac{\widehat{g_0}(S)+\widehat{g_1}(S)}{2}}_p^2 + (p-1) \norm{\frac{\widehat{g_0}(S)-\widehat{g_1}(S)}{2}}_p^2 \right)
 \right)^{1/2} =
\left( \sum_{S \subseteq [n+1]} (p-1)^{|S|} \normb{\widehat{f}(S)}_p^2 \right)^{1/2}
$$
where we used
$\widehat{f}(S) = \half(\widehat{g_0}(S) + \widehat{g_1}(S)) $ and
$\widehat{f}(S \cup \set{n+1}) = \half(\widehat{g_0}(S) - \widehat{g_1}(S))$ for any $S \subseteq [n]$.
\end{proof}

\section{Bounds for $k$-out-of-$n$ quantum random access codes}\label{seckrac}

In this section we prove Theorem~\ref{thm:strongqrac}.
Recall that a $k$-out-of-$n$ random access code allows us to encode $n$
bits into $m$ qubits, such that we can recover any $k$-bit substring with
probability at least $p$. We now define this notion formally. In fact,
we consider a somewhat weaker notion where we only measure the success
probability for a random $k$ subset, and a random input $x \in \{0,1\}^n$.
Since we only prove impossibility results, this clearly makes our results stronger.

\begin{definition}
A $k$-out-of-$n$ quantum random access code on $m$ qubits with success probability $p$
(for short $(k,n,m,p)$-QRAC),
is a map
$$
f:\01^n\to\mathbb{C}^{2^m\times 2^m}
$$
that assigns an $m$-qubit density matrix $f(x)$ to every $x\in\01^n$, and
a quantum measurement $\{M_{S,z}\}_{z\in\01^k}$ to every set
$S\in{[n]\choose k}$, with the property that
$$
\Exp_{x,S}[\Tr(M_{S,x_S}\cdot f(x))] \geq p,
$$
where the expectation is taken over a uniform choice of $x \in \{0,1\}^n$ and
$S \in \ksubsets$, and $x_S$ denotes the $k$-bit substring of $x$ specified by $S$.
\end{definition}

In order to prove Theorem~\ref{thm:strongqrac}, we introduce another notion of QRAC,
which we call \emph{XOR-QRAC}. Here, the goal is to predict the XOR
of the $k$ bits indexed by $S$ (as opposed to guessing all the bits in $S$).
Since one can always predict a bit with probability $\frac{1}{2}$, it is
convenient to define the \emph{bias} of the prediction as $\eps=2p-1$ where $p$ is the probability
of a correct prediction. Hence a bias of $1$ means that the prediction
is always correct, whereas a bias of $-1$ means that it is always wrong.
The advantage of dealing with an XOR-QRAC is that it is easy to express
the best achievable prediction bias without any need to introduce measurements.
Namely, if $f:\01^n\to\mathbb{C}^{2^m\times 2^m}$ is the encoding function,
then the best achievable bias in predicting the XOR of the bits in $S$ (over
a random $\{0,1\}^n$) is exactly half the trace distance between the
average of $f(x)$ over all $x$ with the XOR of the bits in $S$ being $0$
and the average of $f(x)$ over all $x$ with the XOR of the bits in $S$ being $1$.
Using our notation for Fourier coefficients, this can be written simply as
$ \big\| \widehat{f}(S) \big\|_{\rm tr}.$

\begin{definition}\label{def:xorqrac}
A $k$-out-of-$n$ XOR quantum random access code on $m$ qubits with bias $\eps$ (for short $(k,n,m,\eps)$-XOR-QRAC),
is a map
$$
f:\01^n\to\mathbb{C}^{2^m\times 2^m}
$$
that assigns an $m$-qubit density matrix $f(x)$ to every $x\in\01^n$
and has the property that
$$
\Exp_{S \sim \ksubsets}\left[ \trnormb{\widehat{f}(S)} \right] \ge \eps.
$$
\end{definition}

Our new hypercontractive inequality allows us to easily derive the following key lemma:

\begin{lemma} \label{lem:xor-bias}
Let $f:\01^n\to\mathbb{C}^{2^m\times 2^m}$ be any mapping from $n$-bit strings
to $m$-qubit density matrices. Then for any $0\le \delta \le 1$, we have
$$
\sum_{S \subseteq [n]} \delta^{|S|} \trnormb{\widehat{f}(S)}^2 \le 2^{2\delta m}.
$$
\end{lemma}

\begin{proof}
Let $p=1+\delta$. On one hand, by Theorem \ref{thm:main-thm} and Eq.~\eqref{clm:dm-max-norm} we have
$$ \sum_{S \subseteq [n]} (p-1)^{|S|} \normb{\widehat{f}(S)}_p^2  \le
   \bigg( \frac{1}{2^{n}} \sum_{x \in \01^n} \normb{f(x)}_p^p \bigg)^{2/p}
   \le
  \left( \frac{1}{2^{n}} \cdot 2^n \cdot \frac{1}{2^{m}} \right)^{2/p} = 2^{-2m/p}.
$$
On the other hand, by norm monotonicity we have
\begin{eqnarray*}
  \sum_{S \subseteq [n]} (p-1)^{|S|} \normb{\widehat{f}(S)}_p^2   &\ge&
  \sum_{S \subseteq [n]} (p-1)^{|S|} \normb{\widehat{f}(S)}_1^2 =
  2^{-2m} \sum_{S \subseteq [n]} (p-1)^{|S|} \trnormb{\widehat{f}(S)}^2.
\end{eqnarray*}
By rearranging we have
$$
\sum_{S \subseteq [n]} (p-1)^{|S|} \trnormb{\widehat{f}(S)}^2 \le 2^{2m (1-1/p)} \le 2^{2m (p-1)},
$$
as required.
\end{proof}

The following is our main theorem regarding XOR-QRAC.
In particular it shows that if $k=o(n)$ and $m/n<1/(2\ln 2)\approx 0.721$, then the bias will be exponentially small in $k$.

\begin{theorem}\label{xor-bias-thm}
For any $(k,n,m,\eps)$-XOR-QRAC we have the following bound on the bias
$$
\eps \le \left( \frac{(2e\ln 2)m}{k} \right)^{k/2} {\binom{n}{k}}^{-1/2}.
$$
In particular, for any $\eta> 2 \ln 2$ there exists a constant $C_\eta$ such
that if $n/k$ is large enough then for any $(k,n,m,\eps)$-XOR-QRAC,
$$
\eps \le C_\eta \left(\frac{\eta m}{n}\right)^{k/2}.
$$
\end{theorem}

\begin{proof}
Apply Lemma~\ref{lem:xor-bias} with $\delta = \frac{k}{(2 \ln 2)m}$
and only take the sum on $S$ with $|S|=k$. This gives
$$
 \Exp_{S \sim \ksubsets} \left[\trnormb{\widehat{f}(S)}^2 \right] \le 2^{2\delta m } \delta^{-k} {\binom{n}{k}}^{-1}
 = \left( \frac{(2e\ln 2)m}{k} \right)^k {\binom{n}{k}}^{-1}.
$$
The first bound on $\eps$ now follows by convexity (Jensen's inequality).
To derive the second bound, approximate ${n\choose k}$ using Stirling's approximation $n! = \Theta(\sqrt{n} (n/e)^n)$:
$$
{n\choose k}=\frac{n!}{k!(n-k)!}=\Theta\left(\sqrt{\frac{n}{k(n-k)}}\left(\frac{n}{k}\right)^k\left(1+\frac{k}{n-k}\right)^{n-k}\right).
$$
Now use the fact that for large enough $n/k$ we have $(1+k/(n-k))^{(n-k)/k} > (2 e \ln 2)/ \eta$,
and notice that the factor $\sqrt{n/k(n-k)} \ge \sqrt{1/k}$ can be absorbed by this approximation.
\end{proof}

We now derive Theorem~\ref{thm:strongqrac} from Theorem~\ref{xor-bias-thm}.

\begin{proof}[ of Theorem~\ref{thm:strongqrac}]
Consider a $(k,n,m,p)$-QRAC, given by encoding function $f$ and measurements
$\{M_{T,z}\}_{z\in\01^k}$ for all $T\in{[n]\choose k}$.
Define $p_T(w)=\Exp_x\left[\Pr[z\oplus x_T=w]\right]$ as the distribution on the ``error vector''
$w\in\01^k$ of the measurement outcome $z\in\01^k$ when applying $\{M_{T,z}\}$.
By definition, we have that $p\leq \Exp_T[p_T(0^k)]$.

Now suppose we want to predict the parity of the bits of some set $S$ of size at most $k$.
We can do this as follows: uniformly pick a set $T\in{[n]\choose k}$ that contains $S$,
measure $f(x)$ with $\{M_{T,z}\}$, and output the parity of the bits corresponding to $S$ in the measurement outcome $z$.
Note that our output is correct if and only if the bits corresponding to $S$ in the error vector $w$ have even parity.
Hence the bias of our output is
$$
\beta_S=\Exp_{T:T\supseteq S}\left[\sum_{w\in\01^k} p_T(w)\chi_S(w)\right]
  =  2^k \, \Exp_{T:T\supseteq S}\left[\widehat{p_T}(S)\right].
$$
(We slightly abuse notation here by viewing $S$ both as a subset of $T$ and as a subset
of $[k]$ obtained by identifying $T$ with $[k]$.)
Notice that $\beta_S$ can be upper bounded by the best-achievable bias $\trnormb{\widehat{f}(S)}$.

Consider the distribution ${\cal S}$ on sets $S$ defined as follows:
first pick $j$ from the binomial distribution $B(k,1/2)$ and then uniformly pick $S\in{[n]\choose j}$.
Notice that the distribution on pairs $(S,T)$ obtained by first choosing $S \sim {\cal S}$ and then choosing
a uniform $T \supseteq S$ from ${[n]\choose k}$ is identical to the one obtained
by first choosing uniformly $T$ from ${[n]\choose k}$ and then choosing a uniform $S \subseteq T$.
This allows us to show that the average bias $\beta_S$ over $S \sim {\cal S}$ is at least $p$,
as follows:
\begin{align*}
\Exp_{S\sim{\cal S}}\left[\beta_S\right] &=
  2^k\Exp_{S\sim{\cal S},T\supseteq S}\left[\widehat{p_T}(S)\right]\\
  & =  2^k\Exp_{T\sim{[n]\choose k},S\subseteq T}\left[\widehat{p_T}(S)\right]\\
  & =  \Exp_{T\sim{[n]\choose k}}\Bigg[\sum_{S\subseteq T}\widehat{p_T}(S)\Bigg]\\
  & = \Exp_{T\sim{[n]\choose k}}\left[p_T(0^k)\right]  \ge p,
\end{align*}
where the last equality follows from Eq.~\eqref{eq:invfourier}.
On the other hand, using Theorem~\ref{xor-bias-thm} we obtain
\begin{align*}
\Exp_{S\sim{\cal S}}\left[\beta_S\right]
  & \leq  \Exp_{S\sim{\cal S}}\left[\trnormb{\widehat{f}(S)}\right]\\
  & =  \frac{1}{2^k}\sum_{j=0}^k\binom{k}{j}\Exp_{S\sim{[n]\choose j}}\left[\trnormb{\widehat{f}(S)}\right]\\
  & \leq  \frac{1}{2^k}\sum_{j=0}^k\binom{k}{j} C_\eta \left(\frac{\eta m}{n}\right)^{j/2}\\
  & =   C_\eta \left(\frac{1}{2}+\frac{1}{2}\sqrt{\frac{\eta m}{n}}\right)^k,
\end{align*}
where the last equality uses the binomial theorem. Combining the two inequalities completes the proof.
\end{proof}

\section{Direct product theorem for one-way quantum communication}\label{secdirectproddisj}

The setting of communication complexity is by now well-known, so we will not give formal definitions
of protocols etc., referring to~\cite{kushilevitz&nisan:cc,wolf:qccsurvey} instead.
Consider the $n$-bit Disjointness problem in 2-party communication complexity.
Alice receives $n$-bit string $x$ and Bob receives $n$-bit string $y$.
They interpret these strings as subsets of $[n]$ and want to decide
whether their sets are disjoint. In other words, $\DISJ_n(x,y)=1$ if and only if $x\cap y=\emptyset$.
Let $\DISJ_n^{(k)}$ denote $k$ independent instances of this problem.
That is, Alice's input is a $k$-tuple $x_1,\ldots,x_k$ of $n$-bit strings,
Bob's input is a $k$-tuple $y_1,\ldots,y_k$, and they should output all $k$ bits:
$\DISJ_n^{(k)}(x_1,\ldots,x_k,y_1,\ldots,y_k)=\DISJ_n(x_1,y_1),\ldots,\DISJ_n(x_k,y_k)$.
The trivial protocol where Alice sends all her inputs to Bob has success probability~1
and communication complexity $kn$.  We want to show that if the total one-way communication
is much smaller than $kn$ qubits, then the success probability is exponentially
small in $k$. We will do that by deriving a random access code from the protocol's message.

\begin{lemma}
Let $\ell\leq k$.
If there is a $c$-qubit one-way communication protocol
for $\DISJ_n^{(k)}$ with success probability $\sigma$,
then there is an $\ell$-out-of-$kn$ quantum random access code of $c$
qubits with success probability $p\geq\sigma\left(1-\ell/k\right)^{\ell}$.
\end{lemma}

\begin{proof}
Consider the following one-way communication setting: Alice has a $kn$-bit string $x$,
and Bob has $\ell$ distinct indices $i_1,\ldots,i_{\ell}\in[kn]$ chosen uniformly from ${[kn] \choose \ell}$
and wants to learn the corresponding bits of~$x$.

In order to do this, Alice sends the $c$-qubit message corresponding to input $x$ in the $\DISJ_n^{(k)}$ protocol.
We view $x$ as consisting of $k$ disjoint blocks of $n$ bits each.
The probability (over the choice of Bob's input) that $i_1,\ldots,i_{\ell}\in[kn]$ are in $\ell$ different blocks is
$$
\prod_{i=0}^{\ell-1}\frac{kn-in}{kn-i}
\geq \left(\frac{kn-\ell n}{kn}\right)^{\ell}
=\left(1-\frac{\ell}{k}\right)^{\ell}.
$$
If this is the case, Bob chooses his Disjointness inputs $y_1,\ldots,y_k$ as follows.
If index $i_j$ is somewhere in block $b\in[k]$, then he chooses $y_b$ to be the string
having a 1 at the position where $i_j$ is, and 0s elsewhere.
Note that the correct output for the $b$-th instance of Disjointness with inputs
$x$ and $y_1,\ldots,y_k$ is exactly $1-x_{i_j}$.
Now Bob completes the protocol and gets a $k$-bit output for the $k$-fold Disjointness problem.
A correct output tells him the $\ell$ bits he wants to know
(he can just disregard the outcomes of the other $k-\ell$ instances).
Overall the success probability is at least $\sigma(1-\ell/k)^{\ell}$.
Therefore, the random access code that encodes $x$ by Alice's message proves the lemma.
\end{proof}

Combining the previous lemma with our earlier upper bound on $p$ for $\ell$-out-of-$kn$ quantum random
access codes (Theorem~\ref{thm:strongqrac}), we obtain the following upper bound on
the success probability $\sigma$ of $c$-qubit one-way communication protocols for $\DISJ_n^{(k)}$.
For every $\eta> 2\ln 2$ there exists a constant $C_\eta$ such that:
$$
\sigma\leq 2p(1-\ell/k)^{-\ell}
\leq 2C_\eta \left(\left(\frac{1}{2}+\frac{1}{2}\sqrt{\frac{\eta(c+O(k+\log (kn)))}{kn}}\right)\left(\frac{k}{k-\ell}\right)\right)^{\ell}.
$$
Choosing $\ell$ a sufficiently small constant fraction of $k$ (depending on $\eta$),
we obtain a strong direct product theorem for one-way communication:

\begin{theorem}\label{thdptdisjoneway}
For any $\eta>2\ln 2$ the following holds:
for any large enough $n$ and any $k$, every one-way quantum protocol for $\DISJ_n^{(k)}$
that communicates $c\leq kn/\eta$ qubits, has success probability $\sigma\leq 2^{-\Omega(k)}$
(where the constant in the $\Omega(\cdot)$ depends on $\eta$).
\end{theorem}

The above strong direct product theorem (SDPT)
bounds the success probability for protocols that are required to
compute \emph{all} $k$ instances correctly. We call this a \emph{zero-error} SDPT.
What if we settle for a weaker notion of ``success'', namely getting a $(1-\eps)$-fraction of the $k$ instances right,
for some small $\eps>0$? An \emph{$\eps$-error SDPT} is a theorem to the effect
that even in this case the success probability is exponentially small.
An $\eps$-error SDPT follows from a zero-error SDPT as follows.
Run an $\eps$-error protocol with success probability $p$ (``success'' now means
getting $1-\eps$ of the $k$ instances right), guess up to $\eps k$ positions
and change them. With probability at least $p$, the number of errors of
the $\eps$-error protocol is at most $\eps k$, and with probability at least
$1/\sum_{i=0}^{\eps k}{k\choose i}$ we now have corrected all those errors.
Since $\sum_{i=0}^{\eps k}{k\choose i}\leq 2^{k H(\eps)}$ (see, e.g.,~\cite[Corollary~23.6]{jukna:excom}),
we have a protocol that computes all instances correctly with success probability
$\sigma\geq p 2^{-k H(\eps)}$. If we have a zero-error SDPT that bounds
$\sigma\leq 2^{-\gamma k}$ for some $\gamma>H(\eps)$, then it follows that $p$ must
be exponentially small as well: $p\leq 2^{-(\gamma-H(\eps))k}$.
Hence Theorem~\ref{thdptdisjoneway} implies:

\begin{theorem}\label{thepsdptdisjoneway}
For any $\eta>2\ln 2$ there exists an $\eps>0$ such that the following holds:
for every one-way quantum protocol for $\DISJ_n^{(k)}$ that communicates $c\leq kn/\eta$ qubits,
its probability to compute at least a $(1-\eps)$-fraction of the $k$ instances correctly
is at most $2^{-\Omega(k)}$.
\end{theorem}

\section{Lower bounds on locally decodable codes}\label{secldc}

When analyzing locally decodable codes, it will be convenient to view bits
as elements of $\{\pm 1\}$ instead of $\01$.
Formally, a locally decodable code is defined as follows.

\begin{definition}\label{defldc}
$C:\{\pm 1\}^n\rightarrow\{\pm 1\}^N$ is a \emph{$(q,\delta,\eps)$-locally decodable code}
(LDC) if there is a randomized decoding algorithm $A$ such that
\begin{enumerate}
\item For all $x\in\{\pm 1\}^n$, $i\in[n]$, and $y\in\{\pm 1\}^N$ with Hamming distance
$d(C(x),y)\leq\delta N$, we have $\Pr[A^y(i)=x_i]\geq 1/2+\eps$. Here $A^y(i)$ is the
random variable that is $A$'s output given input $i$ and oracle $y$.
\item $A$ makes at most $q$ queries to $y$, non-adaptively.
\end{enumerate}
\end{definition}

In Appendix~\ref{appspecialldc} we show that such a code implies the following:
For each $i\in[n]$, there is a set $M_i$ of at least $\delta\eps N/q^2$ disjoint tuples, each of at most $q$ elements from $[N]$,
and a sign $a_{i,Q}\in\{\pm 1\}$ for each $Q\in M_i$, such that
$$
\Exp_x[a_{i,Q} x_i\prod_{j\in Q}C(x)_j]\geq\frac{\eps}{2^q},
$$
where the expectation is uniformly over all $x\in\{\pm 1\}^n$.
In other words, the parity of each of the tuples in $M_i$ allows us to predict $x_i$ with non-trivial bias (averaged over all $x$).

Kerenidis and de Wolf~\cite{kerenidis&wolf:qldcj} used quantum information theory to show the
lower bound $N=2^{\Omega(\delta\eps^2 n})$ on the length of 2-query LDCs.
Using the new hypercontractive inequality, we can prove a similar
lower bound. Our dependence on $\eps$ and $\delta$ is slightly worse,
but can probably be improved by a more careful analysis.

\begin{theorem}
If $C:\{\pm 1\}^n\rightarrow\{\pm 1\}^N$ is a $(2,\delta,\eps)$-LDC, then $N=2^{\Omega(\delta^2 \eps^4 n)}$.
\end{theorem}

\begin{proof}
Define $f(x)$ as the $N\times N$ matrix whose $(i,j)$-entry is $C(x)_iC(x)_j$.
Since $f(x)$ has rank~1 and its $N^2$ entries are all $+1$ or $-1$, its only non-zero singular value is $N$.
Hence $\norm{f(x)}_p^p=N^{p-1}$ for every $x$.

Consider the $N\times N$ matrices $\widehat{f}(\{i\})$ that are the Fourier transform of $f$ at the singleton sets $\{i\}$:
$$
\widehat{f}(\{i\})=\frac{1}{2^n}\sum_{x\in\{\pm 1\}^n}f(x)x_i.
$$
We want to lower bound $\normb{\widehat{f}(\{i\})}_p$.

With the above notation, each set $M_i$ consists of at least $\delta\eps N/4$ disjoint pairs of indices.\footnote{Actually some of the elements of $M_i$ 
may be singletons. Dealing with this is a technicality that we will ignore here in order to simplify the presentation.}
For simplicity assume $M_i=\{(1,2),(3,4),(5,6),\ldots\}$.
The $2\times 2$ submatrix in the upper left corner of $f(x)$ is
$$
\left(\begin{array}{cc}
1 & C(x)_1C(x)_2\\
C(x)_1C(x)_2 & 1\end{array}\right).
$$
Since $(1,2)\in M_i$, we have $\Exp_x[C(x)_1C(x)_2 x_i a_{i,(1,2)}]\in [\eps/4,1]$.
Hence the $2\times 2$ submatrix in the upper left corner of $\widehat{f}(\{i\})$ is
$$
\left(\begin{array}{cc}
0 & a\\
a & 0\end{array}\right)
$$
for some $a$ with $|a| \in[\eps/4,1]$.
The same is true for each of the first $\delta\eps N/4$ $2\times 2$ diagonal blocks of $\widehat{f}(\{i\})$
(each such $2\times 2$ block corresponds to a pair in $M_i$). Let $P$ be the $N\times N$
permutation matrix that swaps rows 1 and~2, swaps rows 3 and~4, etc.
Then the first $\delta\eps N/2$ diagonal entries of $F_i=P\widehat{f}(\{i\})$ all have absolute value in $[\eps/4,1]$. 

The $\norm{\cdot}_p$ norm is \emph{unitarily invariant}:
$\norm{UAV}_p=\norm{A}_p$ for every matrix $A$ and unitaries $U,V$.
Note the following lemma, which is a special case of~\cite[Eq.~(IV.52) on p.~97]{Bhatia:97a}.
We include its proof for completeness.

\begin{lemma}\label{lemnorms}
Let $\norm{\cdot}$ be a unitarily-invariant norm on the set of $d\times d$ complex matrices.
If $A$ is a matrix and ${\rm diag}(A)$ is the matrix obtained from $A$ by setting its off-diagonal entries to 0,
then $\norm{{\rm diag }(A)}\leq\norm{A}$.
\end{lemma}

\begin{proof}
We will step-by-step set the off-diagonal entries of $A$ to 0, without increasing its norm.
We start with the off-diagonal entries in the $d$th row and column.
Let $D_d$ be the diagonal matrix that has $D_{d,d}=-1$ and $D_{i,i}=1$ for $i<d$.
Note that $D_d A D_d$ is the same as $A$, except that the off-diagonal entries of the $d$th row and column
are multiplied by $-1$.
Hence $A'=(A+D_dAD_d)/2$ is the matrix obtained from $A$ by setting those entries to~0
(this doesn't affect the diagonal).
Since $D_d$ is unitary and every norm satisfies the triangle inequality, we have
$$
\norm{A'}=\norm{(A+D_dAD_d)/2}\leq\frac{1}{2}(\norm{A}+\norm{D_dAD_d})=\norm{A}.
$$
In the second step, we can set the off-diagonal entries in the $(d-1)$st row and column of $A'$ to 0,
using the diagonal matrix $D_{d-1}$ which has a $-1$ only on its $(d-1)$st position. Continuing in this manner,
we set all off-diagonal entries of $A$ to zero without affecting its diagonal,
and without increasing its norm.
\end{proof}

Using this lemma, we obtain
$$
\normb{\widehat{f}(\{i\})}_p=\norm{F_i}_p\geq\norm{{\rm diag}(F_i)}_p\geq \left(\frac{1}{N}(\delta\eps N/2)(\eps/4)^p\right)^{1/p}=(\delta\eps/2)^{1/p}\eps/4.
$$
Using the hypercontractive inequality (Theorem~\ref{thm:main-thm}), we have for any $p\in[1,2]$
$$
n(p-1)(\delta\eps/2)^{2/p}(\eps/4)^2\leq \sum_{i=1}^n (p-1)\normb{\widehat{f}(\{i\})}_p^2 \leq \left(\frac{1}{2^n}\sum_x \norm{f(x)}_p^p\right)^{2/p}=N^{2(p-1)/p}.
$$
Choosing $p=1+1/\log N$ and rearranging implies the result.
\end{proof}

Let us elaborate on the similarities and differences between this proof and the quantum proof of~\cite{kerenidis&wolf:qldcj}.
On the one hand, the present proof makes no use of quantum information theory. It only uses the well known version
of LDCs mentioned after Definition~\ref{defldc}, some basic matrix analysis,
and our hypercontractive inequality for matrix-valued functions.
On the other hand, the proof may still be viewed as a translation of the original quantum proof to a different language.
The quantum proof defines, for each $x$, a $\log(N)$-qubit state $\ket{\phi(x)}$ which is the uniform superposition over the $N$
indices of the codeword $C(x)$.  It then proceeds in two steps: (1) by viewing the elements of $M_i$ as 2-dimensional projectors
in a quantum measurement of $\ket{\phi(x)}$, we can with good probability
recover the parity $C(x)_jC(x)_k$ for a random element $(j,k)$ of the matching $M_i$.
Since that parity has non-trivial correlation with $x_i$, the states $\ket{\phi(x)}$ form a quantum random access code:
they allow us to recover each $x_i$ with decent probability (averaged over all~$x$);
(2) the quantum proof then invokes Nayak's linear lower bound on the number of qubits of a random access code to conclude
$\log N=\Omega(n)$.
The present proof mimics this quantum proof quite closely: the matrix $f(x)$ is, up to normalization, the density matrix corresponding to the
state $\ket{\phi(x)}$; the fact that matrix $\widehat{f}(\{i\})$ has fairly high norm corresponds to the fact that
the parity produced by the quantum measurement has fairly good correlation with~$x_i$; and finally, our invocation
of Theorem~\ref{thm:main-thm} replaces (but is not identical to) the linear lower bound on quantum random access codes.
We feel that by avoiding any explicit use of quantum information theory,
the new proof holds some promise for potential extensions to codes with $q\geq 3$.

\subsubsection*{Acknowledgments}
This work started while the second author was visiting the group in CWI Amsterdam,
and he would like to thank them for their hospitality.
Part of this work was done while the authors were visiting the Institut Henri Poincar{\'e} in Paris,
as part of the program ``Quantum information, computation and complexity'',
and we would like to thank the organizers for their efforts.  We thank Shiri Artstein, Julia Kempe,
Hartmut Klauck, Robert K\"{o}nig, Assaf Naor, Ashwin Nayak, Ryan O'Donnell, Renato Renner,
Alex Samorodnitsky, Falk Unger, Emanuele Viola, and Avi Wigderson for useful discussions and comments.
Thanks to Troy Lee for a preliminary version of~\cite{lss:quantumnof}.

\bibliographystyle{plain}

\appendix

\section{3-party NOF communication complexity of Disjointness}\label{app3partynof}

Some of the most interesting open problems in communication complexity arise
in the ``number on the forehead'' (NOF) model of multiparty communication complexity,
with applications ranging from bounds on proof systems to circuit lower bounds.
Here, there are $\ell$ players and $\ell$ inputs $x_1,\ldots,x_\ell$.
The players want to compute some function $f(x_1,\ldots,x_\ell)$.
Each player~$j$ sees all inputs \emph{except} $x_j$.
In the $\ell$-party version of the Disjointness problem, the $\ell$ players want to figure out
whether there is an index $i\in[n]$ where all $\ell$ input strings have a 1.
For any constant $\ell$, the best known upper bound is linear in $n$~\cite{Grolmusz94}.

While the case $\ell=2$ has been well-understood for a long time,
the first polynomial lower bounds for $\ell\geq 3$ were shown only very recently.
Lee and Shraibman~\cite{lee&shraibman:nofdisj}, and independently Chattopadhyay and Ada~\cite{chattopadhyay&ada:disj},
showed lower bounds of the form $\Omega(n^{1/(\ell+1)})$ on the classical communication complexity for constant $\ell$.
This becomes $\Omega(n^{1/4})$ for $\ell=3$ players.

Stronger lower bounds can be shown if we limit the kind of interaction allowed between the players.
Viola and Wigderson~\cite{viola&wigderson:multipartypj} showed a lower bound of $\Omega(n^{1/(\ell-1)})$
for the \emph{one-way} complexity of $\ell$-player Disjointness, for any constant $\ell$.
In particular, this gives $\Omega(\sqrt{n})$ for $\ell=3$.%
\footnote{Actually, this bound for the case $\ell=3$ was already known earlier; see \cite{bhk:missingbit}.}
An intermediate model was studied by Beame et al.~\cite{bpsw:corruptionj}, namely
protocols where Charlie first sends a message to Bob, and then Alice and Bob are allowed
two-way communication between each other to compute $\DISJ_n(x_1,x_2,x_3)$.
This model is weaker than full interaction, but stronger than the one-way model.
Beame et al.\ showed (using a direct product theorem)
that any protocol of this form requires $\Omega(n^{1/3})$ bits of communication.%
\footnote{Their conference paper had an $\Omega(n^{1/3}/\log n)$ bound, but the journal
version~\cite{bpsw:corruptionj} managed to get rid of the $\log n$.}

Here we strengthen these two 3-player results to \emph{quantum} communication complexity, while at the same
time slightly simplifying the proofs.
These results will follow easily from two direct product theorems: the one for two-way communication
from~\cite{ksw:dpt}, and the new one for one-way communication that we prove here.
Lee, Schechtman, and Shraibman~\cite{lss:quantumnof} have recently extended their $\Omega(n^{1/(\ell+1)})$ 
classical lower bound to $\ell$-player quantum protocols.  While that result holds for a stronger communication model than ours
(arbitrary point-to-point quantum messages), their bound for $\ell=3$ is weaker than ours ($\Omega(n^{1/4})$ vs $\Omega(n^{1/3})$).

\subsection{Communication-type $C\to(B\leftrightarrow A)$}

Consider 3-party Disjointness on inputs $x,y,z\in\01^n$.
Here Alice sees $x$ and $z$, Bob sees $y$ and $z$,
and Charlie sees $x$ and $y$. Their goal is to decide if there is
an $i\in[n]$ such that $x_i=y_i=z_i=1$.

Suppose we have a 3-party protocol $P$ for Disjointness with the following ``flow'' of communication.
Charlie sends a message of $c_1$ classical bits to Alice and Bob (or just to Bob, it doesn't really matter),
who then exchange $c_2$ \emph{qu}bits and compute Disjointness with bounded error probability.
Our lower bound approach is similar to the one of Beame et~al.~\cite{bpsw:corruptionj},
the main change being our use of stronger direct product theorems.
Combining the (0-error) two-way quantum strong direct product theorem for Disjointness from~\cite{ksw:dpt}
with the argument from the end of our Section~\ref{secdirectproddisj}, we have the following $\eps$-error
strong direct product theorem for $k$ instances of 2-party Disjointness:

\begin{theorem}\label{thepsdptdisjtwoway}
There exist constants $\eps>0$ and $\alpha>0$ such that the following holds:
for every two-way quantum protocol for $\DISJ_n^{(k)}$ that communicates at most $\alpha k\sqrt{n}$ qubits,
its probability to compute at least an $(1-\eps)$-fraction of the $k$ instances correctly,
is at most $2^{-\Omega(k)}$.
\end{theorem}

Assume without loss of generality that the error probability of our initial 3-party
protocol $P$ is at most half the $\eps$ of Theorem~\ref{thepsdptdisjtwoway}.
View the $n$-bit inputs of protocol $P$ as consisting of $t$ consecutive blocks of $n/t$ bits each.
We will restrict attention to inputs $z=z_1\ldots z_t$ where one $z_i$ is
all-1, and the other $z_j$ are all-0. Note that for such a $z$, we have $\DISJ_n(x,y,z)=\DISJ_{n/t}(x_i,y_i)$.
Fixing $z$ thus reduces the 3-party Disjointness on $(x,y,z)$
to 2-party Disjointness on a smaller instance $(x_i,y_i)$.
Since Charlie does not see input $z$, his $c_1$-bit message is independent of $z$.
Now by going over all $t$ possible $z$'s,
and running their 2-party protocol $t$ times starting from Charlie's message,
Alice and Bob obtain a protocol $P'$ that computes $t$ independent instances of 2-party Disjointness,
namely on each of the $t$ inputs $(x_1,y_1),\ldots,(x_t,y_t)$.
This $P'$ uses at most $t c_2$ qubits of communication.
For every $x$ and $y$, it follows from linearity of expectation that the expected
number of instances where $P'$ errs, is at most $\eps t/2$
(expectation taken over Charlie's message, and the $t$-fold Alice-Bob protocol).
Hence by Markov's inequality, the probability that $P'$ errs on more than $\eps t$ instances,
is at most 1/2.  Then for every $x,y$ there exists a $c_1$-bit message $m_{xy}$ such that $P'$,
when given that message to start with, with probability at least 1/2 correctly computes $1-\eps$ of all $t$ instances.

Now replace Charlie's $c_1$-bit message by a uniformly random message $m$. Alice and Bob can just generate
this by themselves using shared randomness. This gives a new 2-party protocol $P''$.
For each $x,y$, with probability $2^{-c_1}$ we have $m=m_{xy}$, hence with probability at least
$\frac{1}{2}2^{-c_1}$ the protocol $P''$ correctly
computes $1-\eps$ of all $t$ instances of Disjointness on $n/t$ bits each.
Choosing $t=O(c_1)$ and invoking Theorem~\ref{thepsdptdisjtwoway}
gives a lower bound on the communication in $P''$: $t c_2=\Omega(t\sqrt{n/t})$.
Hence $c_2=\Omega(\sqrt{n/c_1})$.
The overall communication of the original 3-party protocol $P$ is
$$
c_1+c_2=c_1+\Omega(\sqrt{n/c_1})=\Omega(n^{1/3})
$$
(the minimizing value is $t=n^{1/3}$).

This generalizes the bound of  Beame et~al.~\cite{bpsw:corruptionj} to the case where
we allow Alice and Bob to send each other qubits.
Note that this bound is tight for our restricted set of $z$'s, since
Alice and Bob know $z$ and can compute the 2-party Disjointness on the relevant $(x_i,y_i)$
in $O(\sqrt{n^{2/3}})=O(n^{1/3})$ qubits of two-way communication without help from Charlie,
using the optimal quantum protocol for 2-party Disjointness~\cite{aaronson&ambainis:search}.

\subsection{Communication-type $C\to B\to A$}

Now consider an even more restricted type of communication: Charlie sends a classical message to Bob,
then Bob sends a quantum message to Alice, and Alice computes the output.
We can use a similar argument as before, dividing the inputs into $t=O(n^{1/2})$ equal-sized blocks
instead of $O(n^{1/3})$ equal-sized blocks.
If we now replace the two-way SDPT (Theorem~\ref{thepsdptdisjtwoway}) by the new one-way SDPT
(Theorem~\ref{thepsdptdisjoneway}), we obtain a lower bound of $\Omega(\sqrt{n})$
for 3-party bounded-error protocols for Disjointness of this restricted type.

\paragraph{Remark.}
If Charlie's message is quantum as well, then the same approach works, except we need to
reduce the error of the protocol to $\ll 1/t$ at a multiplicative cost of $O(\log t)=O(\log n)$
to both $c_1$ and $c_2$ (Charlie's one quantum message needs to be reused $t$ times).
This worsens the two communication lower bounds to $\Omega(n^{1/3}/\log n)$ and $\Omega(\sqrt{n}/\log n)$ qubits, respectively.

\section{Massaging locally decodable codes to a special form}\label{appspecialldc}

In this appendix we justify the special decoding-format of LDCs claimed after Definition~\ref{defldc}.
First, it will be convenient to switch to the notion of a \emph{smooth code}, introduced by Katz and Trevisan~\cite{katz&trevisan:ldc}.

\begin{definition}
$C:\{\pm 1\}^n\rightarrow\{\pm 1\}^N$ is a $(q,c,\eps)$-\emph{smooth
code} if there is a randomized decoding algorithm $A$ such that
\begin{enumerate}
\item $A$ makes at most $q$ queries, non-adaptively.
\item For all $x\in\{\pm 1\}^n$ and $i\in[n]$ we have
$\Pr[A^{C(x)}(i)=x_i]\geq 1/2+\eps$.
\item For all $x\in\{\pm 1\}^n$, $i\in[n]$, and $j\in[N]$, the probability that on input
$i$ algorithm $A$ queries index $j$ is at most $c/N$.
\end{enumerate}
\end{definition}

Note that smooth codes only require good decoding
on codewords $C(x)$, not on $y$ that are close to $C(x)$.
Katz and Trevisan~\cite[Theorem~1]{katz&trevisan:ldc} established
the following connection:

\begin{theorem}[\cite{katz&trevisan:ldc}]\label{ldctosmooth}
A $(q,\delta,\eps)$-LDC is a $(q,q/\delta,\eps)$-smooth code.
\end{theorem}
\begin{proof}
Let $C$ be a $(q,\delta,\eps)$-LDC and $A$ be its $q$-query decoder.
For each $i\in[n]$, let $p_i(j)$ be the probability that on input~$i$,
algorithm $A$ queries index $j$. Let $H_i=\{j\mid p_i(j)>q/(\delta N)\}$.
Then $|H_i|\leq\delta N$, because $A$ makes no more than $q$ queries.
Let $B$ be the decoder that simulates $A$, except that on input $i$ it does
not make queries to $j\in H_i$, but instead acts as if those bits of its oracle are 0.
Then $B$ does not query any $j$ with probability greater than $q/(\delta N)$.
Also, $B$'s behavior on input $i$ and oracle $C(x)$ is the same as $A$'s behavior
on input $i$ and the oracle $y$ that is obtained by setting the $H_i$-indices of $C(x)$ to 0.
Since $y$ has distance at most $|H_i|\leq\delta N$ from $C(x)$,
we have $\Pr[B^{C(x)}(i)=x_i]=\Pr[A^y(i)=x_i]\geq 1/2+\eps$.
\end{proof}

A converse to Theorem~\ref{ldctosmooth} also holds: a
$(q,c,\eps)$-smooth code is a $(q,\delta,\eps-c\delta)$-LDC,
because the probability that the decoder queries one of
$\delta N$ corrupted positions is at most $(c/N)(\delta N)=c\delta$.
Hence LDCs and smooth codes are essentially equivalent, for
appropriate choices of the parameters.

\begin{theorem}[\cite{katz&trevisan:ldc}]\label{thmatching}
Suppose $C:\{\pm 1\}^n\rightarrow\{\pm 1\}^N$ is a $(q,c,\eps)$-smooth code.
Then for every $i\in[n]$, there exists a set $M_i$, consisting of at least $\eps N/(cq)$ disjoint sets of at most $q$ elements of $[N]$ each,
such that for every $Q\in M_i$ there exists a function $f_Q:\{\pm 1\}^{|Q|}\rightarrow\{\pm 1\}$ with the property
$$
\Exp_x[f_Q(C(x)_Q)x_i]\geq\eps.
$$
Here $C(x)_Q$ is the restriction of $C(x)$ to the bits in $Q$, and the expectation is uniform over all $x\in\{\pm 1\}^n$.
\end{theorem}

\begin{proof}
Fix some $i\in[n]$.
Without loss of generality we assume that to decode $x_i$, the decoder picks some set $Q\subseteq[N]$ (of at most $q$ indices) with probability $p(Q)$,
queries those bits, and then outputs a \emph{random variable} (not yet a function) $f_Q(C(x)_Q)\in\{\pm 1\}$ that depends on the query-answers.
Call such a $Q$ ``good'' if
$$
\Pr_x[f_Q(C(x)_Q)=x_i]\geq 1/2+\eps/2.
$$
Equivalently, $Q$ is good if
$$
\Exp_x[f_Q(C(x)_Q)x_i]\geq\eps.
$$
Now consider the hypergraph $H_i=(V,E_i)$ with vertex-set $V=[N]$ and edge-set $E_i$ consisting of all good sets~$Q$.
The probability that the decoder queries some $Q\in E_i$ is $p(E_i):=\sum_{Q\in E_i} p(Q)$.
If it queries some $Q\in E_i$ then $\Exp_x[f_Q(C(x)_Q)x_i]\leq 1$, and if it queries some $Q\not\in E_i$ then
$\Exp_x[f_Q(C(x)_Q)x_i]<\eps$.
Since the overall probability of outputting $x_i$ is at least $1/2+\eps$ for every $x$, we have
$$
2\eps \leq \Exp_{x,Q}[f_Q(C(x)_Q)x_i] < p(E_i)\cdot 1 + (1-p(E_i))\eps=\eps+p(E_i)(1-\eps),
$$
hence
$$
p(E_i)>\eps/(1-\eps)\geq \eps.
$$
Since $C$ is smooth, for every $j\in[N]$ we have
$$
\sum_{Q\in E_i:j\in Q} p(Q)\leq \sum_{Q:j\in Q} p(Q)=\Pr[A\mbox{ queries }j]\leq \frac{c}{N}.
$$
A \emph{matching} of $H_i$ is a set of disjoint $Q\in E_i$. Let $M_i$ be a matching in $H_i$ of maximal size.
Our goal is to show $|M_i|\geq \eps N/(cq)$. Define $T=\cup_{Q\in M_i} Q$.
This set $T$ has at most $q|M_i|$ elements, and intersects each $Q\in E_i$ (otherwise $M_i$ would not be maximal).
We now lower bound the size of $M_i$ as follows:
$$
\eps< p(E_i)=\sum_{Q:Q\in E_i} p(Q)\stackrel{(*)}{\leq}\sum_{j\in T}\sum_{Q\in E_i:j\in Q}p(Q)\leq \frac{c|T|}{N}\leq \frac{cq|M_i|}{N},
$$
where $(*)$ holds because each $Q\in E_i$ is counted exactly once on the left and at least once on the right (since $T$ intersects each $Q\in E_i$).
Hence $|M_i|\geq \eps N/(cq)$. It remains to turn the random variables $f_Q(C(x)_Q)$ into fixed values in $\{\pm 1\}$;
it is easy to see that this can always be done without reducing the correlation $\Exp_x[f_Q(C(x)_Q)x_i]$.
\end{proof}

The previous theorem establishes that the decoder can just pick a uniformly random element $Q\in M_i$,
and then continue as the original decoder would on those queries, at the expense of reducing the average
success probability by a factor 2.  In principle, the decoder could output any function of the $|Q|$ queried bits that it wants.
We now show (along the lines of~\cite[Lemma~2]{kerenidis&wolf:qldcj}) that we can restrict attention to parities
(or their negations), at the expense of decreasing the average success probability by another factor of $2^q$.

\begin{theorem}\label{thmatchingandparity}
Suppose $C:\{\pm 1\}^n\rightarrow\{\pm 1\}^N$ is a $(q,c,\eps)$-smooth code.
Then for every $i\in[n]$, there exists a set $M_i$, consisting of at least $\eps N/(cq)$ disjoint sets of at most $q$ elements of $[N]$ each,
such that for every $Q\in M_i$ there exists an $a_{i,Q}\in\{\pm 1\}$ with the property that
$$
\Exp_x[a_{i,Q} x_i\prod_{j\in Q}C(x)_j]\geq\frac{\eps}{2^q}.
$$
\end{theorem}

\begin{proof}
Fix $i\in[n]$ and take the set $M_i$ produced by Theorem~\ref{thmatching}.
For every $Q\in M_i$ we have
$$
\Exp_x[f_Q(C(x)_Q)x_i]\geq\eps.
$$
We would like to turn the functions $f_Q:\{\pm 1\}^{|Q|}\rightarrow\{\pm 1\}$ into parity functions.
Consider the Fourier transform of $f_Q$:
for $S\subseteq[|Q|]$ and $z\in\{\pm 1\}^{|Q|}$, define parity function $\chi_S(z)=\prod_{j\in S} z_j$ and
Fourier coefficient $\widehat{f_Q}(S)=\frac{1}{2^{|Q|}}\sum_zf_Q(z)\chi_S(z)$.
Then we can write
$$
f_Q=\sum_S\widehat{f_Q}(S)\chi_S.
$$
Using that $\widehat{f_Q}(S)\in[-1,1]$ for all $S$, we have
$$
\eps \leq \Exp_x[f_Q(C(x)_Q)x_i]
 = \sum_S \widehat{f_Q}(S)\Exp_x[x_i\chi_S(C(x)_Q)]
\leq  \sum_S \left| \Exp_x[x_i\chi_S(C(x)_Q)] \right|.
$$
Since the right-hand side is the sum of $2^{|Q|}$ terms, there exists an $S$ with
$\displaystyle |\Exp_x[x_i\chi_S(C(x)_Q)]|\geq \frac{\eps}{2^{|Q|}}$.

Defining $a_{i,Q}=\mbox{sign}(\Exp_x[x_i\chi_S(C(x)_Q)])\in\{\pm 1\}$, we have
$$
\Exp_x[a_{i,Q} x_i\prod_{j\in S}C(x)_j]= |\Exp_x[x_i\chi_S(C(x)_Q)]|\geq \frac{\eps}{2^{|Q|}}\geq \frac{\eps}{2^q}.
$$
The theorem follows by replacing each $Q$ in $M_i$ by the set $S$ just obtained from it.
\end{proof}

Combining Theorems~\ref{ldctosmooth} and~\ref{thmatchingandparity} gives the
decoding-format claimed after Definition~\ref{defldc}.

\end{document}